\documentclass[journal]{IEEEtran}
\usepackage[T1]{fontenc}
\usepackage[utf8]{inputenc}
\usepackage{amsmath}
\usepackage{amssymb}
\usepackage{amsthm}
\usepackage{graphicx}
\usepackage{algorithm}
\usepackage{algorithmic}
\usepackage{cite}
\usepackage[unicode=true]{hyperref}
\hypersetup{colorlinks=true, allcolors=blue}

\theoremstyle{plain}
\newtheorem{assumption}{Assumption}
\newtheorem{lemma}{Lemma}
\newtheorem{theorem}{Theorem}
\newtheorem{corollary}{Corollary}
\newtheorem{proposition}{Proposition}
\newtheorem{fact}{Fact}

\begin{document}
\title{Token-Level Attribution of Multi-Mode LLM Watermarks: Optimal Rule and Belief Weighting}
\author{Rui~Zhou, Licheng~Zhao, Chao~Shen, and Wenqiang~Pu%
\thanks{A preliminary version of this work was submitted to the 2027 IEEE International Conference on Acoustics, Speech and Signal Processing (ICASSP)~\cite{zhou2027icassp}.}
\thanks{Rui Zhou is with the City University of Hong Kong (Dongguan), Dongguan, China, and also with the City University of Hong Kong, Hong Kong (email: \href{mailto:rui.zhou@cityu-dg.edu.cn}{rui.zhou@cityu-dg.edu.cn}).}
\thanks{Licheng Zhao, Chao Shen, and Wenqiang Pu are with the Shenzhen Research Institute of Big Data, The Chinese University of Hong Kong (Shenzhen), Shenzhen, China (e-mail: \href{mailto:zhaolicheng@sribd.cn}{zhaolicheng@sribd.cn}; \href{mailto:chaoshen@sribd.cn}{chaoshen@sribd.cn}; \href{mailto:wpu@sribd.cn}{wpu@sribd.cn}).}}
\markboth{}{}
\maketitle

\begin{abstract}
A watermark verifier may hold many modes, i.e., secret keys, and a document may mix spans of different modes with unwatermarked text. Directly reporting the mode with the largest pivotal statistic (pivot) inflates the Type I error, and restoring the nominal level reduces power as candidate modes are added. 
We formulate token-level multi-mode attribution for the Gumbel-max watermark as maximizing correct-label power under a Type I error constraint given a mode prior. We show that the optimal rule reports the mode with the largest prior-weighted alternative density and rejects when this density exceeds one common threshold.
Based on this rule, we propose the Belief-weighted Likelihood Attribution with Density Estimation (BLADE) detector. It estimates the mode prior from neighboring pivots and the alternative density from reference token probabilities. Under a conditional product model of the pivots, it keeps the exact Type I error.
We further show that BLADE approaches the optimal rule as its prior and density estimates improve, with a sharp coefficient on the density error.
On OPT-1.3B and other open-weight large language models, BLADE attains higher correct-label power than the compared detectors across candidate counts, token-substitution rates, generation temperatures and generator models.

\end{abstract}

\begin{IEEEkeywords}
Large language models, watermarking, multi-mode attribution, optimal detector, neighboring evidence.
\end{IEEEkeywords}

\section{Introduction}
\label{sec:intro}

\IEEEPARstart{W}{atermarking} embeds a key-dependent statistical signal into the text generated by a large language model (LLM). A verifier holding the secret key can then test whether a given text was generated by the model \cite{kirchenbauer2023,aaronson2023,kuditipudi2023}. Existing schemes either bias the next-token prediction (NTP) distribution toward a key-dependent subset of the vocabulary, e.g., the green-list watermark \cite{kirchenbauer2023,zhao2024provable}, or preserve it in expectation over the pseudorandomness \cite{aaronson2023,kuditipudi2023,hu2024unbiased,wu2024resilient,dathathri2024scalable}. The latter class includes the Gumbel-max watermark \cite{aaronson2023,li2025,fernandez2023three}. For the Gumbel-max watermark, the detector computes at each token a pivotal statistic (pivot), namely the keyed pseudorandom value of the observed token. It then tests whether the pivot follows the uniform distribution of unwatermarked text or the stochastically larger distribution of watermarked text \cite{li2025}.
These schemes and their detectors, however, mostly consider a single key and make a binary decision about a whole text.

Many practical applications require a finer answer. A model provider may assign different keys to different users or, for a single user, to different products or usage states, e.g., polishing a human draft versus authoring a manuscript from scratch. We refer to each key, together with the user, product or state that it identifies, as a watermark mode. The text submitted to a verifier may then mix unwatermarked spans, as in \cite{li2024edits,li2025optimal}, with spans generated under different modes. The unwatermarked spans may be written by a human, generated by the model without a watermark, or generated with a watermark that is no longer retained after later editing. We call such a text a mixed document.
The verifier must decide, at each token, whether a watermark is retained and, if so, under which of $K$ candidate modes it was generated. We refer to this problem as multi-mode attribution. Since a detection with a wrong mode label is an attribution error, the relevant performance measure is the correct-label power, i.e., the probability of detecting a retained watermark and assigning it to the correct mode.

A direct approach computes one pivot per candidate mode and reports the mode with the largest pivot whenever that pivot exceeds a threshold. We call this rule Max-Pivot. However, testing several candidates at once inflates the overall Type I error above the nominal level of each individual test, since each candidate pivot may exceed the threshold under the null hypothesis.
Restoring the overall nominal level requires a stricter threshold for each candidate, i.e., a multiplicity correction, and we call the resulting rule calibrated Max-Pivot. This threshold increases as candidates are added, lowering the probability that the pivot of the true mode exceeds it.
In addition, a correct detection requires the pivot of the true mode to exceed the pivots of all other modes, which also becomes less likely as candidates are added. Both effects grow with the number of candidates, because the candidate modes share one null hypothesis and one Type I error budget at every token.
We refer to these effects as the competition among modes. This competition characterizes multi-mode attribution and is absent from single-mode detection.

Several methods use multiple keys or identify the source of a text, but they are not designed for the multi-mode attribution problem considered here. Multi-user watermarks embed a user-specific signal to trace a text to an individual user or to members of a colluding group \cite{cohen2025many}. Watermark-based attribution assigns each user a distinct watermark selected to improve the accuracy of identifying the user \cite{jiang2026attribution}. Randomized key selection uses several keys and accepts a text when exactly one of them is detected \cite{aremu2026}. Multi-bit watermarks embed a message, e.g., a user identifier, and recover it by aggregating evidence across many tokens \cite{yoo2024advancing}. Methods that localize watermarked segments in mixed-source text \cite{segment2024,zhao2025,pan2024waterseeker,bonnerjee2025wiser,blanchet2026} label positions within mixed documents but consider a single key. Most of the multi-key methods above assume that the whole text carries only one watermark, and even the collusion-resistant scheme \cite{cohen2025many} identifies users rather than labeling tokens. The multi-key methods therefore determine which user or key produced a whole text, whereas multi-mode attribution must determine which mode, if any, generated each token.

Within a source span, neighboring tokens usually share a mode, so the neighbors of a token provide evidence for identifying its mode more accurately. In our preliminary work \cite{zhou2027icassp}, we proposed a Belief-weighted (BW) detector for the Gumbel-max watermark. It allocates the Type I error budget across the candidates with weights given by mode probabilities from a hidden Markov model of local persistence \cite{rabiner1989}. On mixed documents, it improved the correct-label power over calibrated Max-Pivot.
However, the BW detector and its weights were designed heuristically rather than derived from an optimality criterion.

In this paper, we cast multi-mode attribution as a constrained optimization problem over the decision rule, in which the neighboring evidence enters through the mode prior, i.e., the prior probability of each mode at the target token. Under a conditional product model of the pivots, the optimal rule has an explicit form. The main contributions are as follows:

\begin{itemize}
\item We formulate token-level multi-mode attribution for the Gumbel-max watermark as maximizing the prior-weighted correct-label power under a Type I error constraint given the mode prior. We show that the optimal rule reports the mode with the largest prior-weighted alternative density and rejects when this density exceeds one common threshold. Under the uniform product null of the pivots, this threshold splits the Type I error budget among the modes and also bounds the misattribution rate.
\item Based on the optimal rule, we propose the Belief-weighted Likelihood Attribution with Density Estimation (BLADE) detector. It estimates the required alternative density from reference token probabilities and the mode prior with a working Markov chain that includes a null state. Under the conditional product model, it keeps the exact Type I error for any such estimates formed from the other positions and independent reference data.
\item We bound the loss of BLADE relative to the optimum in terms of the prior and density estimation errors, and show that the coefficient of the density error is sharp.
\item We evaluate the detectors on OPT-1.3B and other open-weight LLMs. BLADE attains higher correct-label power than calibrated Max-Pivot, the Local-Count weighting of \cite{zhou2027icassp} and BW across candidate counts, token-substitution rates, generation temperatures and generator models. It also keeps the Type I error near the nominal level on human-written and edited text.
\end{itemize}

This paper is organized as follows. Section~\ref{sec:prelim} reviews single-mode Gumbel-max watermarking and its pivot distribution. Section~\ref{sec:formulation} formulates multi-mode attribution, and Section~\ref{sec:optimal} derives the optimal rule. Section~\ref{sec:detectors} constructs BW and BLADE from the optimal rule and bounds BLADE's loss from estimated inputs. Section~\ref{sec:experiments} presents numerical experiments, and Section~\ref{sec:conclusion} concludes the paper.

\emph{Notation:} Boldface lowercase and uppercase letters denote vectors and matrices, respectively. $\Pr(\cdot)$ and $\mathbb{E}[\cdot]$ denote probability and expectation, $\operatorname{Unif}(0,1)$ is the uniform distribution on $[0,1]$, and $\mathbb{I}\{\cdot\}$ is the indicator function. $\mathbf 1_K$ and $\mathbf I_K$ are the $K$-dimensional all-ones vector and identity matrix, $(\cdot)^{\top}$ denotes transposition, $\odot$ is the elementwise product, and $\operatorname{norm}(\mathbf x)=\mathbf x/\sum_ix_i$ for a nonnegative vector $\mathbf x$ with a positive sum. For $\mathbf{x}\in\mathbb{R}^K$, $\lVert\mathbf{x}\rVert_1=\sum_j\lvert x_j\rvert$, and for an integrable function $g$ on $[0,1]$, $\lVert g\rVert_{L^1}=\int_0^1\lvert g(x)\rvert\,dx$.
\section{Preliminaries: Single-Mode Gumbel-Max Watermarking}
\label{sec:prelim}

\subsection{Gumbel-Max Watermarking and the Single-Token Test}
\label{sec:prelim-gumbel}

The Gumbel-max watermark \cite{aaronson2023,li2025} is specified by a secret key $\kappa$ shared by the model provider and the verifier \cite{kuditipudi2023,li2024edits}. At position $t$, the LLM computes an NTP distribution $\mathbf{p}_t=(p_{t,w})_{w\in\mathcal{V}}$ over a vocabulary $\mathcal{V}$, after any temperature scaling or truncation. A keyed hash of the context window $C_t$, e.g., the preceding $m\ge1$ tokens with $m=5$ in Section~\ref{sec:experiments}, defines a mapping $\mathcal A_\kappa$. This mapping assigns each token $w\in\mathcal V$ a pseudorandom value $U_{t,w}=\mathcal A_\kappa(w\mid C_t)\in[0,1]$, and the next token is
\begin{equation}
\label{eq:gumbel-max}
X_t=\arg\max_{w\in\mathcal{V}:\,p_{t,w}>0}\ \frac{\log U_{t,w}}{p_{t,w}}.
\end{equation}
The statistical properties of the watermark depend on the distribution of the pseudorandom values $U_{t,w}$. We adopt the following assumption, which extends the standard one \cite{aaronson2023,li2025} to multiple keys.

\begin{assumption}
\label{ass:pseudorandom}
The pseudorandom values $U_{t,w}$ are independent $\operatorname{Unif}(0,1)$ variables across vocabulary entries, keys and distinct context windows.
\end{assumption}

At positions with a previously unseen context window \cite{cai2025collision}, Assumption~\ref{ass:pseudorandom} makes \eqref{eq:gumbel-max} an exact sample from $\mathbf{p}_t$, so the watermark preserves the NTP distribution \cite{aaronson2023,li2025}.
To detect the presence of the watermark, the verifier evaluates the keyed mapping at the selected token to obtain the pivot \cite{li2025}
\begin{equation}
\label{eq:pivot}
r_t=\mathcal{A}_\kappa(X_t\mid C_t)=U_{t,X_t}\in[0,1].
\end{equation}

If the token is chosen independently of the pseudorandom values, e.g., in human-written text, then $r_t\sim F_0$ with cumulative distribution function (CDF) $F_0(x)=x$. If the token is generated by \eqref{eq:gumbel-max}, $r_t$ follows a stochastically larger distribution $F_H$ with a nondecreasing density. Each position thus poses the one-sided test
\begin{equation}
\label{eq:single-test}
\begin{cases}
\mathcal{H}_0^{(t)}&:\ r_t\sim F_0,\\
\mathcal{H}_1^{(t)}&:\ r_t\sim F_H.
\end{cases}
\end{equation}
Under $\mathcal{H}_0^{(t)}$, $1-r_t$ is an exact $p$-value. Since the likelihood ratio of $F_H$ to $F_0$ is nondecreasing, the most powerful level-$\alpha$ test \cite{lehmann2005} is the single-mode rule
\begin{equation}
\label{eq:single-rule}
\begin{cases}
\mathcal{H}_0^{(t)}, & \text{if }r_t\le1-\alpha,\\[2pt]
\mathcal{H}_1^{(t)}, & \text{if }r_t>1-\alpha,
\end{cases}
\end{equation}
with power $1-F_H(1-\alpha)$.

\subsection{Alternative Distribution of the Pivot}
\label{sec:prelim-signal}
Under $\mathcal{H}_1^{(t)}$, the pivot depends on the generator only through the NTP probability of the selected token, $p_t^\star=p_{t,X_t}$. Let $d_t=1/p_t^\star$. The following lemma restates a result of \cite{li2025}, and the subsequent corollary builds on the mixture distribution of \cite{fernandez2023three,li2025}.

\begin{lemma}[{[\citen{li2025}]}]
\label{lem:pivot-law}
Under Assumption~\ref{ass:pseudorandom}, if the watermark generated $X_t$ at a position with a previously unseen context window, then the conditional distribution of the pivot given $p_t^\star$ is
\begin{equation}
\label{eq:pivot-law}
\Pr\bigl(r_t\le x\mid p_t^\star\bigr)=x^{d_t},\qquad0\le x\le1.
\end{equation}
\end{lemma}

\begin{corollary}
\label{cor:signal-law}
Under the conditions of Lemma~\ref{lem:pivot-law}, the CDF and density of a watermarked pivot, averaged over the distribution of $p_t^\star$, are
\begin{equation}
\label{eq:signal-law}
F_H(x)=\mathbb{E}\bigl[x^{d_t}\bigr],\qquad f_H(x)=\mathbb{E}\bigl[d_t\,x^{d_t-1}\bigr].
\end{equation}
Moreover, $f_H$ is nondecreasing and continuous, $f_H(1)=\mathbb E[d_t]\le\lvert\mathcal V\rvert$, and $f_H$ is strictly increasing whenever $\Pr(p_t^\star<1)>0$.
\end{corollary}

The bound on $f_H(1)$ holds because $\mathbb E[d_t\mid\mathbf p_t]$ is the support size of $\mathbf p_t$, and continuity follows by dominated convergence. Since $F_H$ depends on the generator only through the distribution of $p_t^\star$, the verifier can learn it offline (Section~\ref{sec:blade}). The model $F_H(x)=x^d$ of \cite{zhou2027icassp} is the special case $p_t^\star\equiv1/d$, which Section~\ref{sec:form-maxpivot} uses to quantify the competition cost.

By \eqref{eq:pivot-law}, a token with $p_t^\star$ close to $1$ has a nearly uniform pivot and is hard to detect, whereas a token with a low $p_t^\star$ has a pivot concentrated near $1$. The single-token power $\mathbb E[1-(1-\alpha)^{d_t}]$ thus grows as the text becomes less predictable. Fig.~\ref{fig:single-mode} shows empirical pivot CDFs on OPT-1.3B text for unwatermarked tokens and for watermarked tokens generated at three temperatures. At temperature $1.0$ and $\alpha=0.05$, rule \eqref{eq:single-rule} detects only $42.0\%$ of the watermarked tokens. Higher temperatures typically lower $p_t^\star$, which explains the order of the three watermarked curves.

\begin{figure}[!t]
\centering
\includegraphics{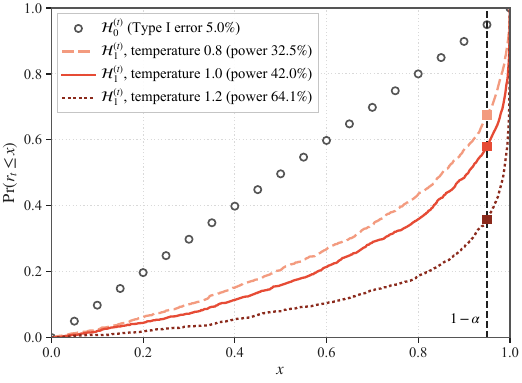}
\caption{Empirical pivot CDFs on OPT-1.3B text.}
\label{fig:single-mode}
\end{figure}

\section{Problem Formulation}
\label{sec:formulation}

\subsection{Token-Level Hypothesis Test}
\label{sec:form-multimode}

In multi-mode watermarking, the provider holds $K$ watermark modes, e.g., one per user, and mode $j$ generates text by \eqref{eq:gumbel-max} with its own key $\kappa_j$. At most one mode acts at each position. A document may concatenate spans generated under different modes with unwatermarked spans. The active mode is latent but locally persistent, i.e., adjacent tokens usually share it. The watermark also need not be retained at every token of a watermarked span, e.g., after later editing. Because modes can change within a document and the watermark may be lost at some tokens, we formulate attribution at the token level. 

The verifier does not know which mode acted, so it computes a pivot for every candidate mode
\begin{equation}
\label{eq:mode-pivots}
r_{t,j}=\mathcal{A}_{\kappa_j}(X_t\mid C_t)\in[0,1],\qquad j=1,\dots,K,
\end{equation}
and forms the vector $\mathbf{r}_t=(r_{t,1},\dots,r_{t,K})$. At position $t$, the attribution problem is the multiple hypothesis test \cite{zhou2027icassp}
\begin{equation}
\label{eq:multi-test}
\begin{aligned}
\mathcal{H}_0^{(t)}&:\ r_{t,i}\sim F_0\ \text{i.i.d.},\quad i=1,\dots,K,\\
\mathcal{H}_j^{(t)}&:\ r_{t,j}\sim F_{H,j},\quad r_{t,i}\sim F_0\ \ (i\neq j),
\end{aligned}
\end{equation}
where $\mathcal{H}_0^{(t)}$ states that no watermark is retained at position $t$, because either no mode acted or the coupling between the token and its pseudorandom values was lost. For $j=1,\dots,K$, $\mathcal{H}_j^{(t)}$ states that mode $j$ acted and its watermark is retained. Thus a token whose watermark is not retained is a null even inside a watermarked span, and its pivots are uniform. We denote by $Z_t\in\{0,1,\dots,K\}$ the true state, so that $Z_t=j$ under $\mathcal{H}_j^{(t)}$. 

By Assumption~\ref{ass:pseudorandom}, the $K$ pivots in \eqref{eq:multi-test} are independent at a position whose context window has not occurred earlier. Repeated context windows \cite{cai2025collision} are identified from the observed text, and the verifier masks such positions and reports them as unassessed. Throughout, $t$ denotes an unmasked position. We write $f_j$ for the density of $F_{H,j}$. It is nondecreasing as in Section~\ref{sec:prelim-signal} and may differ across modes.

Unlike the single-mode test \eqref{eq:single-test}, the $K$ modes compete with each other in the multi-mode test \eqref{eq:multi-test}. A decision rule $\delta_t\in\{0,1,\dots,K\}$ maps the pivots $\mathbf r_t$ to the null decision ($\delta_t=0$) or a mode label. Its correct-label power for mode $j$, $P_j(\delta_t)=\Pr(\delta_t=j\mid\mathcal{H}_j^{(t)})$, counts a successful detection only if it reports the true mode, and its Type I error is $P_{\mathrm I}(\delta_t)=\Pr(\delta_t\neq0\mid\mathcal{H}_0^{(t)})$. Since the pivots are independent, these quantities are
\begin{equation}
\label{eq:power-typeI}
\begin{aligned}
P_j(\delta_t)&=\int_{[0,1]^K}\mathbb{I}\{\delta_t(\mathbf{r})=j\}\,f_j(r_j)\,d\mathbf{r},\\
P_{\mathrm I}(\delta_t)&=\int_{[0,1]^K}\mathbb{I}\{\delta_t(\mathbf{r})\neq0\}\,d\mathbf{r}.
\end{aligned}
\end{equation}

\subsection{Competition Cost of Max-Pivot}
\label{sec:form-maxpivot}

A direct approach applies the single-mode rule \eqref{eq:single-rule} to the largest of the $K$ pivots and reports its mode, which we call Max-Pivot:
\begin{equation}
\label{eq:max-pivot}
\begin{cases}
\mathcal{H}_0^{(t)}, & \text{if }\max_ir_{t,i}\le1-\alpha,\\[2pt]
\mathcal{H}_j^{(t)}, & \text{if }j=\arg\max_ir_{t,i}\text{ and }r_{t,j}>1-\alpha.
\end{cases}
\end{equation}
Equivalently, it runs $K$ separate single-mode tests, declares a watermark retained if any of them rejects the null, and reports the mode with the largest pivot. This rule is simple, but it incurs a competition cost: its Type I error grows with $K$, and its correct-label power falls. The following two facts quantify this cost under a simple alternative distribution.
\begin{fact}
\label{fact:max-pivot}
Under Assumption~\ref{ass:pseudorandom}, let $F_{H,j}(x)=x^d$ for all $j$ and some $d\ge1$, i.e., $p_t^\star\equiv1/d$. When $K$ is large, the Type I error of Max-Pivot approaches $1$, and its correct-label power is $P_j(\delta_t)\approx d/K$.
\end{fact}
\begin{IEEEproof}
See Appendix~\ref{app:C}.
\end{IEEEproof}
\begin{fact}
\label{fact:cal-max-pivot}
Under the conditions of Fact~\ref{fact:max-pivot}, calibrated Max-Pivot, which replaces the threshold $1-\alpha$ by $(1-\alpha)^{1/K}$, has Type I error $\alpha$. When $K$ is large, its correct-label power is $P_j(\delta_t)\approx\alpha d/K$.
\end{fact}
\begin{IEEEproof}
See Appendix~\ref{app:C}.
\end{IEEEproof}

Because the true pivot must exceed the other $K-1$ pivots, the correct-label power decays as $1/K$ with or without calibration, and restoring the level $\alpha$ costs a further factor $\alpha$. For example, at $\alpha=0.05$, the Type I error of Max-Pivot grows from $9.8\%$ at $K=2$ to $80.6\%$ at $K=32$. With $d=4$, calibrated Max-Pivot retains $26\%$ of the power of a single-mode test that knows the generating mode at $K=4$ and only $3.4\%$ at $K=32$, in line with $\alpha d/K$. The loss stems from missing mode information rather than from the pivot itself, since the single-mode test with the known mode does not incur it.

\subsection{Prior-Weighted Attribution Problem}
\label{sec:form-problem}

Max-Pivot treats all modes alike, although the modes need not be equally likely at position $t$. We assume $\Pr(Z_t\neq0)>0$ and define the mode prior
\begin{equation}
\label{eq:mode-prior}
q_{t,j}=\Pr\bigl(Z_t=j\mid Z_t\neq0\bigr),\qquad j=1,\dots,K,
\end{equation}
as the probability that a watermark retained at position $t$ comes from mode $j$. A rule that uses $\mathbf q_t=(q_{t,1},\dots,q_{t,K})$ can direct its Type I error budget toward the likely modes. We therefore seek the rule that maximizes the prior-weighted correct-label power under a token-level Type I error constraint:
\begin{equation}
\label{eq:problem}
\begin{aligned}
\max_{\delta_t}\quad&\sum_{j=1}^Kq_{t,j}\,P_j(\delta_t)\\
\text{s.t.}\quad&P_{\mathrm I}(\delta_t)\le\alpha.
\end{aligned}
\end{equation}
By \eqref{eq:mode-prior}, the objective equals $\Pr(\delta_t=Z_t\mid Z_t\neq0)$, and the constraint bounds the Type I error at position $t$.

\section{Optimal Attribution Rule}
\label{sec:optimal}

\subsection{Structure of the Optimal Rule}
\label{sec:opt-rule}

Problem \eqref{eq:problem} weighs the correct-label powers of the $K$ modes against a single Type I error budget. Let $g_t(\mathbf r_t)=\max_jq_{t,j}f_j(r_{t,j})$ denote the largest prior-weighted alternative density at position $t$. The following theorem shows that the optimal rule thresholds this score.

\begin{theorem}
\label{thm:optimal}
Under Assumption~\ref{ass:pseudorandom}, let $0<\alpha<1$ and suppose every $f_j$ is continuous and strictly increasing. An optimal decision rule for problem \eqref{eq:problem} is
\begin{equation}
\label{eq:opt-rule}
\delta_t^\star=\begin{cases}
0, & g_t(\mathbf r_t)\le c_t,\\[2pt]
\arg\max_jq_{t,j}f_j(r_{t,j}), & \text{otherwise},
\end{cases}
\end{equation}
where $c_t>0$ is the unique constant such that $P_{\mathrm I}(\delta_t^\star)=\alpha$.
\end{theorem}
\begin{IEEEproof}
See Appendix~\ref{app:A}.
\end{IEEEproof}

By Corollary~\ref{cor:signal-law}, the conditions of Theorem~\ref{thm:optimal} hold for the Gumbel-max watermark: each $f_j$ is continuous, and it is strictly increasing whenever $p_t^\star<1$ with positive probability under mode $j$.

Theorem~\ref{thm:optimal} keeps the simple threshold form of Max-Pivot and draws its gain from the score alone. Max-Pivot \eqref{eq:max-pivot} thresholds the largest pivot $\max_jr_{t,j}$, whereas rule \eqref{eq:opt-rule} thresholds the largest prior-weighted alternative density $g_t(\mathbf r_t)$ and reports the mode that attains it. This structure agrees with the optimal joint detection and estimation test under a Type I error constraint \cite{moustakides2011}, in which the mode label plays the role of the estimated parameter \cite{middleton1968}. The watermark model adds the uniform product null in \eqref{eq:multi-test}. This null gives the null CDF of the score in the product form \eqref{eq:null-cdf}, which yields the budget split and the misattribution bound of Sections~\ref{sec:opt-calibration} and~\ref{sec:opt-misattribution}. The following corollary shows that rule \eqref{eq:opt-rule} recovers two familiar rules as special cases of the prior.

\begin{corollary}
\label{cor:benchmarks}
Under the conditions of Theorem~\ref{thm:optimal}:
\begin{enumerate}
\item[(i)] if $\mathbf q_t=\mathbf 1_K/K$ and $f_j=f$ for all $j$, rule \eqref{eq:opt-rule} is calibrated Max-Pivot of Fact~\ref{fact:cal-max-pivot};
\item[(ii)] if $q_{t,j}=1$ for the known generating mode $j$, rule \eqref{eq:opt-rule} reduces to applying the single-mode test \eqref{eq:single-rule} to $r_{t,j}$, which we call the Oracle.
\end{enumerate}
\end{corollary}

\begin{IEEEproof}
The proof follows by substituting the two priors into \eqref{eq:opt-rule} and is omitted.
\end{IEEEproof}

The two cases are the extremes of mode information. Without information that distinguishes the modes, no mode can be favored, and the prior is uniform. A common alternative density then makes the modes interchangeable. With full information, the generating mode is known, and the prior is a point mass. Thus the power gap between calibrated Max-Pivot and the Oracle measures the potential gain from mode information.

\subsection{Calibration and Budget Allocation}
\label{sec:opt-calibration}

For a continuous, strictly increasing $f_j$, we extend its inverse $f_j^{-1}$ to $[0,\infty)$ by clipping, i.e., $f_j^{-1}(y)=0$ for $y\le f_j(0)$ and $f_j^{-1}(y)=1$ for $y\ge f_j(1)$.
Let $A_t$ denote the null CDF of $g_t(\mathbf r_t)$. By \eqref{eq:multi-test}, the independence of the uniform null pivots gives
\begin{equation}
\label{eq:null-cdf}
\begin{aligned}
A_t(c)&=\Pr\!\left(\max_jq_{t,j}f_j(r_{t,j})\le c\,\middle|\,\mathcal H_0^{(t)}\right)\\
&=\prod_{j:q_{t,j}>0}f_j^{-1}\!\left(\frac{c}{q_{t,j}}\right),\qquad c\ge0,
\end{aligned}
\end{equation}
where modes with $q_{t,j}=0$ have zero score and are omitted from the product.

At $c=c_t$, the complement of each factor is the per-mode level of the corresponding single-mode test:
\begin{equation}
\label{eq:private-level}
s_{t,j}=1-f_j^{-1}\bigl(c_t/q_{t,j}\bigr),\qquad q_{t,j}>0,
\end{equation}
with $s_{t,j}=0$ when $q_{t,j}=0$. By the monotonicity of $f_j$, rule \eqref{eq:opt-rule} rejects on the union of per-mode upper tails $\{r_{t,j}>1-s_{t,j}\}$. Its Type I error is therefore $1-A_t(c_t)$, and calibration gives
\begin{equation}
\label{eq:balance}
A_t(c_t)=\prod_{j=1}^K(1-s_{t,j})=1-\alpha.
\end{equation}
We refer to \eqref{eq:balance} as the calibration equation. It splits the Type I error budget among the modes: mode $j$ receives the per-mode level $s_{t,j}$, and the levels are linked through the common threshold $c_t$ by \eqref{eq:private-level}. For a common density, a mode with a lower prior receives a smaller level, so the optimal rule shifts the budget toward the more likely modes. In particular, the optimal rule withdraws the budget from unlikely modes: mode $j$ receives no budget if and only if $q_{t,j}f_j(1)\le c_t$. Fig.~\ref{fig:optimal-region} illustrates the resulting rejection region for two modes with a common density.

The decision does not require solving the calibration equation for $c_t$. Let $v_t=g_t(\mathbf r_t)$ denote the observed score.

\begin{lemma}
\label{lem:direct-eval}
If $A_t$ is strictly increasing wherever $0<A_t(c)<1$, then for $0<\alpha<1$, $v_t>c_t$ if and only if $A_t(v_t)>1-\alpha$.
\end{lemma}
\begin{IEEEproof}
By \eqref{eq:balance}, $A_t(c_t)=1-\alpha\in(0,1)$, so the claim follows from the monotonicity of $A_t$ and its strict increase where $0<A_t<1$.
\end{IEEEproof}

By Lemma~\ref{lem:direct-eval}, rule \eqref{eq:opt-rule} rejects the null hypothesis if and only if $A_t(v_t)>1-\alpha$, since the $f_j$ in Theorem~\ref{thm:optimal} make $A_t$ strictly increasing where $0<A_t<1$ (Appendix~\ref{app:A}). The verifier therefore evaluates the null CDF at the observed score and compares it with $1-\alpha$, without computing $c_t$.

\begin{figure}[!t]
\centering
\includegraphics{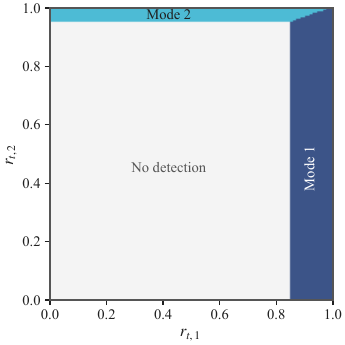}
\caption{Optimal rejection region for $K=2$, $\mathbf{q}_t=(0.7,0.3)$ and $\alpha=0.2$, with the density estimate $\widehat f$ in \eqref{eq:density-est} for both modes.}
\label{fig:optimal-region}
\end{figure}

\subsection{Control of Misattribution}
\label{sec:opt-misattribution}

The per-mode levels also control the misattribution rate $\Pr(\delta_t\notin\{0,j\}\mid\mathcal H_j^{(t)})$, i.e., the probability of reporting a wrong mode for a retained watermark.

\begin{lemma}
\label{lem:misattribution}
Under the conditions of Theorem~\ref{thm:optimal}, rule \eqref{eq:opt-rule} satisfies, for $j=1,\dots,K$,
\begin{equation}
\label{eq:misattribution}
\Pr\bigl(\delta_t^\star\notin\{0,j\}\mid\mathcal H_j^{(t)}\bigr)\le\sum_{i\neq j}s_{t,i}\le-\log(1-\alpha).
\end{equation}
\end{lemma}
\begin{IEEEproof}
Under $\mathcal H_j^{(t)}$, each pivot $r_{t,i}$ with $i\neq j$ is uniform by \eqref{eq:multi-test}, and rule \eqref{eq:opt-rule} reports mode $i$ only if $q_{t,i}f_i(r_{t,i})>c_t$, i.e., $r_{t,i}>1-s_{t,i}$. The union bound gives the first inequality, and $s\le-\log(1-s)$ with \eqref{eq:balance} gives the second.
\end{IEEEproof}

The Type I error budget thus limits wrong-mode reports without a separate constraint in \eqref{eq:problem}. At $\alpha=0.05$, for example, the misattribution rate is at most $5.13\%$.

\section{From the Optimal Rule to Practical Detectors}
\label{sec:detectors}

The optimal rule \eqref{eq:opt-rule} requires the prior $\mathbf q_t$ and the alternative densities $f_j$, which are unknown to the verifier. We first show that replacing them by estimates keeps the exact Type I error, since the neighboring pivots inform the prior without altering the null distribution of the target pivots. We then construct two practical detectors. The Belief-weighted (BW) detector extracts a prior estimate from the neighboring pivots through a working Markov chain \cite{rabiner1989} and replaces the alternative densities by a simple heuristic score. The Belief-weighted Likelihood Attribution with Density Estimation (BLADE) detector refines both components. It extracts the prior estimate through a working chain with a richer state space and estimates the alternative densities from a reference sample of token probabilities. Finally, we bound BLADE's power loss from estimation errors and analyze the complexity of both detectors.

\subsection{Neighboring Evidence and Calibration}
\label{sec:prior-neighbors}
\label{sec:calibration-estimated}

Text generated by one mode tends to be contiguous, e.g., a sentence or a paragraph, so the neighboring pivots carry information about the mode at position $t$. Let $\mathcal R_{-t}$ comprise the pivot vectors $\mathbf r_u$ at the unmasked positions $u\neq t$, together with their positions. Conditioning on $\mathcal R_{-t}$ replaces \eqref{eq:mode-prior} by the conditional prior
\begin{equation}
\label{eq:mode-prior-cond}
q_{t,j}=\Pr\bigl(Z_t=j\mid\mathcal R_{-t},\,Z_t\neq0\bigr),\qquad j=1,\dots,K.
\end{equation}
An estimate of $\mathbf q_t$ from $\mathcal R_{-t}$ makes the threshold $c_t$ depend on the neighbors. The Type I guarantee of Section~\ref{sec:optimal} extends to these neighbor-dependent thresholds when the neighbors carry no information about the null target pivots, as the following assumption states.

\begin{assumption}
\label{ass:product}
Under each hypothesis $\mathcal H_z^{(t)}$, $z=0,1,\dots,K$, given the NTP distribution $\mathbf p_t$ and the token $X_t$, $\mathcal R_{-t}$ is independent of the pseudorandom values $\mathcal A_{\kappa_i}(w\mid C_t)$, $w\in\mathcal V$, $i=1,\dots,K$, at position $t$.
\end{assumption}

The following corollary follows directly from Lemma~\ref{lem:pivot-law}.

\begin{corollary}
\label{cor:product}
Under Assumptions~\ref{ass:pseudorandom} and~\ref{ass:product}, conditionally on $\mathcal R_{-t}$, (a) under $\mathcal H_0^{(t)}$, the pivots $r_{t,1},\dots,r_{t,K}$ are independent $\operatorname{Unif}(0,1)$ variables; and (b) under $\mathcal H_j^{(t)}$, $j=1,\dots,K$, the pivots $r_{t,i}$, $i\neq j$, are independent $\operatorname{Unif}(0,1)$ variables, independent of $r_{t,j}$, and $r_{t,j}$ has the nondecreasing density $f_j(x)=\mathbb E[d_tx^{d_t-1}\mid\mathcal R_{-t},\mathcal H_j^{(t)}]$.
\end{corollary}

By Assumption~\ref{ass:product}, conditioning on $\mathcal R_{-t}$ changes the distribution of $\mathbf p_t$ and $X_t$ but not that of the pseudorandom values given them. Under $\mathcal H_j^{(t)}$, the neighbors thus affect the target pivot $r_{t,j}$ only through $f_j$. We refer to \eqref{eq:multi-test} under Assumption~\ref{ass:product} as the conditional product model.

A practical detector replaces the unknown prior by an estimate $\widehat{\mathbf q}_t$, called its belief, formed from $\mathcal R_{-t}$, and the unknown densities by working scores $\phi_j:[0,1)\to[0,\infty)$. The following lemma shows that these replacements preserve the exact Type~I error.

\begin{lemma}
\label{lem:calibration}
Let $0<\alpha<1$, and let the prior estimates $\widehat q_{t,j}>0$ and the continuous, strictly increasing working scores $\phi_j$ be functions of $\mathcal R_{-t}$ and of reference data independent of the scored document. Under Assumptions~\ref{ass:pseudorandom} and~\ref{ass:product}, rule \eqref{eq:opt-rule} with $\widehat q_{t,j}\phi_j$ in place of $q_{t,j}f_j$, calibrated by \eqref{eq:balance}, satisfies $\Pr(\delta_t\neq0\mid\mathcal R_{-t},\mathcal H_0^{(t)})=\alpha$.
\end{lemma}

\begin{IEEEproof}
The proof follows the derivation of \eqref{eq:balance}, given Corollary~\ref{cor:product}(a) and the monotonicity of the scores, and is omitted.
\end{IEEEproof}

Calibration therefore holds however inaccurate the prior estimates and working scores are. Their errors affect only the power.

From here on, $\mathbf q_t$, $F_{H,j}$, $f_j$, $P_j$ and $P_{\mathrm I}$ are conditional on $\mathcal R_{-t}$, and a decision rule may also depend on $\mathcal R_{-t}$, e.g., through $\mathbf q_t$. In the following, we construct two practical detectors that implement rule \eqref{eq:opt-rule} with estimated priors and working scores.

\subsection{The BW Detector}
\label{sec:bw}

BW implements the decision structure of \eqref{eq:opt-rule} in a simple, heuristic manner. It estimates $\mathbf q_t$ from the neighboring pivots with a working Markov model of local persistence and replaces the unknown alternative density $f_j(r)$ by the fixed inverse-$p$-value score $h(r)=(1-r)^{-1}$.

\subsubsection{Prior Estimate}
\label{sec:bw-prior}

To approximate $\mathbf{q}_t$ from neighboring pivots, BW models the mode of the span containing position $t$ by a working Markov chain $M_t\in\{1,\dots,K\}$ with no null state and transition matrix
\begin{equation}
\label{eq:bw-transition}
\mathbf{P}=\rho_{\mathrm{sw}}\mathbf{I}_K+\frac{1-\rho_{\mathrm{sw}}}{K-1}\bigl(\mathbf{1}_K\mathbf{1}_K^{\top}-\mathbf{I}_K\bigr),
\end{equation}
where $\rho_{\mathrm{sw}}$ is the persistence probability. The observation evidence at position $t$ is $\mathbf l_t=(l_{t,1},\dots,l_{t,K})^{\top}$ with $l_{t,j}=\min\{h(r_{t,j}),\bar l\}$, where the fixed cap $\bar l$ prevents a single extreme pivot from dominating the recursions. At masked positions, $\mathbf l_t=\mathbf1_K$, so state transitions through them are preserved.

The forward--backward recursions \cite{rabiner1989}
\begin{equation}
\label{eq:fb-recursion}
\begin{aligned}
\mathbf{a}_{t+1}&=\mathbf{P}^{\top}\operatorname{norm}(\mathbf{a}_t\odot\mathbf{l}_t),\\
\mathbf{b}_t&=\operatorname{norm}\bigl(\mathbf{P}(\mathbf{l}_t\odot\mathbf{b}_{t+1})\bigr)
\end{aligned}
\end{equation}
are initialized with $\mathbf{a}_1=\mathbf{b}_{n+1}=\mathbf{1}_K/K$ for a document of $n$ tokens. Since $\mathbf a_t$ uses the evidence before $t$ and $\mathbf b_{t+1}$ the evidence after $t$, their product excludes $\mathbf l_t$ and is a function of $\mathcal R_{-t}$, as Lemma~\ref{lem:calibration} requires. Mixing the normalized product with the uniform distribution gives
\begin{equation}
\label{eq:bw-prior}
\widehat{\mathbf{q}}_t=(1-\varepsilon)\operatorname{norm}(\mathbf{a}_t\odot\mathbf{b}_{t+1})+\varepsilon\mathbf{1}_K/K.
\end{equation}
Here $\varepsilon\in(0,1]$ keeps every component of $\widehat{\mathbf q}_t$ positive.

\subsubsection{Working Score}
\label{sec:bw-score}

Under the null, $1-r_{t,j}$ is an exact $p$-value (Section~\ref{sec:prelim-gumbel}), and $h$ is its reciprocal. Like $f_j$, $h$ increases with the pivot, so BW uses it in place of $f_j$ in \eqref{eq:opt-rule}. Since $h$ is not integrable on $[0,1]$, it is a working score rather than a density estimate.

\subsubsection{Decision Rule}

With the prior estimate \eqref{eq:bw-prior} and the working score $h$, BW replaces $q_{t,j}f_j(r_{t,j})$ in \eqref{eq:opt-rule} by $\widehat q_{t,j}h(r_{t,j})=\widehat q_{t,j}/(1-r_{t,j})$. The condition $\widehat q_{t,j}h(r_{t,j})>c_t$ is equivalent to $r_{t,j}>1-\widehat q_{t,j}/c_t$, so \eqref{eq:balance} calibrates the threshold $c_t$ for this score through
\begin{equation}
\label{eq:bw-balance}
\prod_{j=1}^K\left(1-\frac{\widehat q_{t,j}}{c_t}\right)=1-\alpha,\qquad c_t>\max_j\widehat q_{t,j}.
\end{equation}
The solution is unique, since the product increases continuously from $0$ to $1$ on this interval. The per-mode levels are
\begin{equation}
\label{eq:bw-levels}
s_{t,j}=\frac{\widehat q_{t,j}}{c_t},
\end{equation}
which are proportional to the belief. The resulting BW decision rule has the form of \eqref{eq:opt-rule}:
\begin{equation}
\label{eq:bw-rule}
\delta_t=\begin{cases}0, & \max_i\widehat q_{t,i}h(r_{t,i})\le c_t,\\[2pt] \arg\max_j\widehat q_{t,j}h(r_{t,j}), & \text{otherwise}.\end{cases}
\end{equation}

For the observed score $v_t=\max_j\widehat q_{t,j}h(r_{t,j})$, the null CDF $A_t(v_t)=\prod_{j=1}^K(1-\widehat q_{t,j}/v_t)$ is strictly increasing above $\max_j\widehat q_{t,j}$, and $h$ is continuous and strictly increasing on $[0,1)$. By Lemmas~\ref{lem:direct-eval} and~\ref{lem:calibration}, rule \eqref{eq:bw-rule} is conditionally calibrated and rejects when $A_t(v_t)>1-\alpha$, without computing $c_t$ (Algorithm~\ref{alg:bw}).

The choice of $h$ thus gives BW closed-form per-mode levels \eqref{eq:bw-levels} and a closed-form null CDF. The BLADE detector below refines both the prior estimate and the working score to improve the power.

\begin{algorithm}[!t]
\caption{BW Attribution Detector of One Document}
\label{alg:bw}
\begin{algorithmic}[1]
\REQUIRE pivots $\{\mathbf r_t\}_{t=1}^n$ for $K$ candidate modes, repeated-context mask, fixed parameters $(\rho_{\mathrm{sw}},\bar l)$, $\alpha$ and $\varepsilon$.
\STATE Form $\mathbf P$ by \eqref{eq:bw-transition} and initialize $\mathbf a_1=\mathbf b_{n+1}=\mathbf1_K/K$.
\STATE Compute $\mathbf l_t=(l_{t,1},\dots,l_{t,K})^{\top}$ with $l_{t,j}=\min\{(1-r_{t,j})^{-1},\bar l\}$; set $\mathbf l_t=\mathbf1_K$ at masked positions.
\STATE Run the forward--backward recursions \eqref{eq:fb-recursion} .
\FOR{each unmasked $t$}
\STATE Compute $\widehat{\mathbf q}_t$ by \eqref{eq:bw-prior}.
\STATE Compute $v_t=\max_j\widehat q_{t,j}/(1-r_{t,j})$.
\STATE Compute $A_t(v_t)=\prod_{j=1}^K(1-\widehat q_{t,j}/v_t)$.
\STATE If $A_t(v_t)>1-\alpha$, set $\delta_t=\arg\max_j\widehat q_{t,j}/(1-r_{t,j})$; otherwise set $\delta_t=0$.
\ENDFOR
\end{algorithmic}
\end{algorithm}

\subsection{The BLADE Detector}
\label{sec:blade}

BLADE replaces the heuristic score $h$ by an estimate of the alternative density from a reference sample. It estimates $\mathbf q_t$ with a working chain that adds a null state for unwatermarked spans and uses the density estimate as its evidence. We therefore estimate the density first.

\subsubsection{Density Estimate}
\label{sec:blade-density}

Let $\{p_i^\star\}_{i=1}^N$ be a reference sample of selected-token probabilities under the generator's sampling distribution, supplied offline by the model provider independently of the scored document. Replacing the expectation in \eqref{eq:signal-law} by its empirical average gives
\begin{equation}
\label{eq:density-est}
\widehat f(x)=\frac{1}{N}\sum_{i=1}^N d_i x^{d_i-1},\qquad d_i=1/p_i^\star.
\end{equation}
By \eqref{eq:signal-law}, the alternative density depends on the generator only through the distribution of $p_t^\star$, which is set by the model but not by the key. A reference sample from the same generator thus serves all modes, and we use $\widehat f$ as a common working approximation to the conditional densities $f_j$. It is computed once offline and does not depend on $t$ or the scored document.

\subsubsection{Prior Estimate}
\label{sec:blade-prior}

Let $S_t\in\{0,1,\dots,K\}$ denote the mode of the span containing position $t$, with $S_t=0$ for an unwatermarked span. Unlike the BW state $M_t$, $S_t$ includes a null state. A token in a mode-$j$ span has $Z_t=j$ if it retains the watermark and $Z_t=0$ otherwise, and we model retention as independent across positions with probability $\theta\in(0,1)$.

To accommodate unwatermarked spans, BLADE models $S_t$ by a $(K+1)$-state working Markov chain with transition matrix
\begin{equation}
\label{eq:blade-transition}
\begin{aligned}
\bar{\mathbf{P}}&=\begin{bmatrix}\rho_0 & \frac{1-\rho_0}{K}\mathbf{1}_K^{\top}\\[3pt] (1-\rho_1)(1-\beta)\mathbf{1}_K & \mathbf Q\end{bmatrix},\\[3pt]
\mathbf Q&=\rho_1\mathbf{I}_K+\frac{(1-\rho_1)\beta}{K-1}\bigl(\mathbf{1}_K\mathbf{1}_K^{\top}-\mathbf{I}_K\bigr),
\end{aligned}
\end{equation}
where $\mathbf Q$ is the mode-to-mode block. Here $\rho_0\in[0,1]$ and $\rho_1\in(0,1]$ are the persistence probabilities of unwatermarked and watermarked spans, and $\beta\in[0,1]$ is the probability that leaving a mode span leads directly to another mode rather than to the null state.

Under this chain, all pivots at position $t$ are uniform given $S_t=0$. Given $S_t=j\ge1$, the mode-$j$ pivot has the density $(1-\theta)+\theta\widehat f(r_{t,j})$, and the other pivots are uniform. The likelihood ratios relative to the null span state, $\lambda_{t,j}=(1-\theta)+\theta\widehat f(r_{t,j})$, form the evidence vector $\bar{\boldsymbol\lambda}_t=(1,\lambda_{t,1},\dots,\lambda_{t,K})^{\top}$, which replaces $\mathbf l_t$ in the message recursions.

We obtain the $(K+1)$-dimensional messages $\bar{\mathbf a}_t$ and $\bar{\mathbf b}_t$ by using $\bar{\mathbf P}$ and $\bar{\boldsymbol\lambda}_t$ in \eqref{eq:fb-recursion}, with $\bar{\mathbf a}_1=\bar{\mathbf b}_{n+1}=\mathbf1_{K+1}/(K+1)$ and $\bar{\boldsymbol\lambda}_t=\mathbf1_{K+1}$ at masked positions. As in BW, their normalized product $\boldsymbol\gamma_t=\operatorname{norm}(\bar{\mathbf a}_t\odot\bar{\mathbf b}_{t+1})$ excludes the target evidence. Under the working chain, $\gamma_{t,j}=\Pr(S_t=j\mid\mathcal R_{-t})$ for $j=0,1,\dots,K$.

Since retention has the common probability $\theta$ given $S_t$, the working chain gives $\Pr(Z_t=j\mid\mathcal R_{-t})=\theta\gamma_{t,j}$ for $j\ge1$, where $\sum_{i=1}^K\gamma_{t,i}>0$ because the likelihoods are positive and $\rho_1>0$. Conditioning on $Z_t\neq0$ cancels $\theta$, and the uniform mixing of \eqref{eq:bw-prior} yields the estimate of \eqref{eq:mode-prior-cond}
\begin{equation}
\label{eq:blade-prior}
\widehat q_{t,j}=(1-\varepsilon)\frac{\gamma_{t,j}}{\sum_{i=1}^K\gamma_{t,i}}+\frac{\varepsilon}{K},\quad j=1,\dots,K.
\end{equation}

\subsubsection{Decision Rule}
\label{sec:blade-rule}

With the prior estimate \eqref{eq:blade-prior} and the density estimate \eqref{eq:density-est}, BLADE replaces $q_{t,j}f_j(r_{t,j})$ in \eqref{eq:opt-rule} by $\widehat q_{t,j}\widehat f(r_{t,j})$, which gives the rule
\begin{equation}
\label{eq:blade-rule}
\delta_t=\begin{cases}
0, & \max_i\widehat q_{t,i}\widehat f(r_{t,i})\le c_t,\\[2pt]
\arg\max_j\widehat q_{t,j}\widehat f(r_{t,j}), & \text{otherwise},
\end{cases}
\end{equation}
where $c_t$ is determined by \eqref{eq:private-level}--\eqref{eq:balance} using $\widehat{\mathbf q}_t$ and $\widehat f$. The observed maximum score and its conditional null CDF are
\begin{equation}
\label{eq:blade-direct}
v_t=\max_j\widehat q_{t,j}\widehat f(r_{t,j}),\qquad
A_t(v_t)=\prod_{j=1}^K\widehat f^{-1}\!\left(\frac{v_t}{\widehat q_{t,j}}\right),
\end{equation}
where $\widehat f^{-1}$ is the clipped inverse of Section~\ref{sec:opt-calibration}. Since the prior estimates are positive and $\widehat f$ is continuous and strictly increasing whenever some $d_i>1$, Lemmas~\ref{lem:direct-eval} and~\ref{lem:calibration} apply: the calibrated rule \eqref{eq:blade-rule} rejects when $A_t(v_t)>1-\alpha$ and reports the maximizing mode (Algorithm~\ref{alg:blade}). Appendix~\ref{app:E} gives an optional explicit computation of $c_t$ and the per-mode levels.

\begin{algorithm}[!t]
\caption{BLADE Attribution Detector of One Document}
\label{alg:blade}
\begin{algorithmic}[1]
\REQUIRE pivots $\{\mathbf r_t\}_{t=1}^n$ for $K$ candidate modes, repeated-context mask, offline density estimate $\widehat f$ in \eqref{eq:density-est}, fixed parameters $(\theta,\rho_0,\rho_1,\beta)$, $\alpha$ and $\varepsilon$.
\STATE Form $\bar{\mathbf P}$ by \eqref{eq:blade-transition} and initialize $\bar{\mathbf a}_1=\bar{\mathbf b}_{n+1}=\mathbf1_{K+1}/(K+1)$.
\STATE Compute $\bar{\boldsymbol\lambda}_t=(1,\lambda_{t,1},\dots,\lambda_{t,K})^{\top}$ with $\lambda_{t,j}=(1-\theta)+\theta\widehat f(r_{t,j})$; set $\bar{\boldsymbol\lambda}_t=\mathbf1_{K+1}$ at masked positions.
\STATE Run recursions \eqref{eq:fb-recursion} with $\bar{\mathbf P}$ and $\bar{\boldsymbol\lambda}_t$.
\FOR{each unmasked $t$}
\STATE Compute $\boldsymbol\gamma_t=\operatorname{norm}(\bar{\mathbf a}_t\odot\bar{\mathbf b}_{t+1})$ and $\widehat{\mathbf q}_t$ by \eqref{eq:blade-prior}.
\STATE Compute $v_t=\max_j\widehat q_{t,j}\widehat f(r_{t,j})$.
\STATE Compute $A_t(v_t)=\prod_{j=1}^K\widehat f^{-1}(v_t/\widehat q_{t,j})$ with the clipped inverse.
\STATE If $A_t(v_t)>1-\alpha$, set $\delta_t=\arg\max_j\widehat q_{t,j}\widehat f(r_{t,j})$; otherwise set $\delta_t=0$.
\ENDFOR
\end{algorithmic}
\end{algorithm}

\subsection{Effect of Estimation Errors on BLADE}
\label{sec:analysis-blade}

By Lemma~\ref{lem:calibration}, the estimation errors of BLADE affect only its power. By Corollary~\ref{cor:product}, Theorem~\ref{thm:optimal} holds conditionally on $\mathcal R_{-t}$, so we bound this power loss relative to the optimal rule, conditionally on $\mathcal R_{-t}$ and the reference data.

\begin{proposition}
\label{prop:blade-loss}
Under the conditions of Theorem~\ref{thm:optimal} and Assumption~\ref{ass:product}, if some $d_i>1$, the BLADE rule $\delta_t$ in \eqref{eq:blade-rule}, with $c_t$ solving \eqref{eq:balance} exactly, satisfies
\begin{equation}
\label{eq:blade-loss}
\begin{aligned}
0\le{}&\mathrm{OPT}-\sum_{j=1}^Kq_{t,j}P_j(\delta_t)\\
\le{}&\|\mathbf q_t-\widehat{\mathbf q}_t\|_1
+\frac12\sum_{j=1}^Kq_{t,j}\|f_j-\widehat f\|_{L^1}.
\end{aligned}
\end{equation}
Here $\mathrm{OPT}$ denotes the optimal value of \eqref{eq:problem}.
\end{proposition}

\begin{IEEEproof}
See Appendix~\ref{app:B}.
\end{IEEEproof}

The bound separates the two estimated inputs of BLADE: the prior error and the density error weighted by the prior. The prior in \eqref{eq:blade-loss} is the final estimate \eqref{eq:blade-prior}.

To distinguish reference-sample error from mismatch in the alternative distribution, let the reference mixture density $f_{\mathrm{ref}}$ be the population counterpart of \eqref{eq:density-est}. The triangle inequality gives
\begin{equation}
\label{eq:density-split}
\begin{aligned}
\frac12\sum_{j=1}^Kq_{t,j}\lVert f_j-\widehat f\rVert_{L^1}
\le{}&\frac12\lVert f_{\mathrm{ref}}-\widehat f\rVert_{L^1}\\
&+\frac12\sum_{j=1}^Kq_{t,j}\lVert f_j-f_{\mathrm{ref}}\rVert_{L^1}.
\end{aligned}
\end{equation}
For an i.i.d. reference sample, the first term, the reference-sample error, tends to zero almost surely as $N\to\infty$ (Appendix~\ref{app:B}). The second term measures the mismatch between the reference mixture density, which averages over positions, and the conditional densities $f_j$, which condition on the neighbors. If both terms in \eqref{eq:blade-loss} tend to zero, BLADE's correct-label power tends to $\mathrm{OPT}$.

\begin{figure}[!t]
\centering
\includegraphics{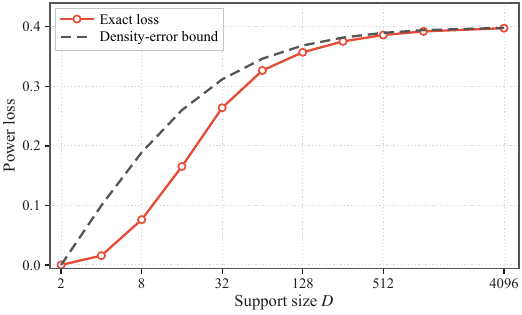}
\caption{Exact power loss of rule \eqref{eq:blade-rule} relative to rule \eqref{eq:opt-rule} and the density-error bound in \eqref{eq:blade-loss} versus $D$, for $K=2$, $\alpha=0.05$ and $\widehat{\mathbf q}_t=\mathbf q_t=(0.4,0.6)$.}
\label{fig:bound-tightness}
\end{figure}

We validate the density term of Proposition~\ref{prop:blade-loss} using a toy example. We set $K=2$, $\alpha=0.05$ and $\widehat{\mathbf q}_t=\mathbf q_t=(0.4,0.6)$, so the prior error is zero. The NTP distributions under modes 1 and 2 are uniform on $D$ and two tokens, respectively, where $D$ is an integer with $2\le D\le\lvert\mathcal V\rvert$, so \eqref{eq:signal-law} gives $f_1(x)=Dx^{D-1}$ and $f_2(x)=2x$. The reference sample is drawn under mode 2, so \eqref{eq:density-est} gives $\widehat f=f_{\mathrm{ref}}=f_2$ exactly for any $N$, and only the mode-1 density error remains. For each $D$, the exact power loss of rule \eqref{eq:blade-rule} relative to rule \eqref{eq:opt-rule} is compared to the density-error bound in \eqref{eq:blade-loss}. As shown in Fig.~\ref{fig:bound-tightness}, the exact loss stays below the bound for every displayed $D$, as Proposition~\ref{prop:blade-loss} guarantees. The bound is also tight, since the two approach each other as $D$ grows and both tend to $0.4$ (Appendix~\ref{app:B}).

\subsection{Computational Complexity}
\label{sec:complexity}

We analyze Algorithms~\ref{alg:bw} and \ref{alg:blade} for a document of $n$ tokens. The structure of \eqref{eq:bw-transition} and \eqref{eq:blade-transition} reduces each matrix--vector product in \eqref{eq:fb-recursion}, including multiplication by the transpose, to one sum over the mode components and componentwise updates, i.e., $O(K)$ operations. With $O(K)$ operations per position for the evidence, scores and null-CDF product, BW has $O(nK)$ time and memory complexity per document.

BLADE also queries $\widehat f$ and its inverse, $2K$ times per position. Since $\widehat f$ does not depend on the document, we tabulate it once at $G$ ordered points, in $O(NG)$ operations reusable across documents. Each query then takes $O(\log G)$ by binary search and linear interpolation, instead of $O(N)$ by \eqref{eq:density-est}. Any continuous, strictly increasing interpolant evaluated exactly, e.g., piecewise-linear interpolation of an increasing table, satisfies Lemmas~\ref{lem:direct-eval} and \ref{lem:calibration}, so calibration remains exact. Appendix~\ref{app:E} details the log-scale tabulation and the bisection used in the experiments. Algorithm~\ref{alg:blade} thus costs $O(nK\log G)$ time and $O(nK+G)$ working memory, so for fixed $G$, both detectors scale linearly in $n$ and $K$.

\section{Numerical Experiments}
\label{sec:experiments}

We evaluate the detectors on mixed text with real pivots for all candidate modes. The experiments examine three claims. First, the detectors keep the Type I error near the nominal level on human-written and edited text. Second, BLADE attains the highest correct-label power among the practical detectors as the candidate count, substitution rate, generation temperature and generator model vary. Third, a mode-weight illustration on one document shows how the detectors use the neighboring evidence.

\subsection{Experimental Protocol}
\label{sec:exp-protocol}
\label{sec:exp-data}
\label{sec:exp-documents}
\label{sec:exp-null}
\label{sec:exp-detectors}
\label{sec:exp-metrics}

Following \cite{kirchenbauer2023,li2024edits,zhou2027icassp}, we generate continuations of 50-token C4 realnewslike prompts \cite{raffel2020} with OPT-1.3B \cite{zhang2022} at temperature $1$. The continuations are watermarked by \eqref{eq:gumbel-max} under four generating modes with a five-token hash context, and $p_t^\star$ is recorded under the sampling distribution. The $224$ sources are split into $32$ development, $24$ calibration and $168$ test sources, each with a watermarked continuation under one mode and an unwatermarked continuation. A layout draws four test sources, one per mode, and no two layouts share a source. Its mixed document of about $400$ tokens interleaves their four $50$-token watermarked spans with five $40$-token spans from their unwatermarked continuations. A fraction of the watermarked tokens is then replaced by tokens drawn uniformly from the vocabulary, as in \cite{li2024edits}, and a position retains the watermark if the six-token sequence ending at it occurs in the original continuation. Repeated contexts are masked as in Section~\ref{sec:detectors}. Type I error is measured on human-written C4 continuations, where $\mathcal H_0^{(t)}$ holds, and on unretained positions of the edited documents, an operational null.

On $42$ layouts, we vary $K\in\{4,8,\dots,32\}$ at $20\%$ substitution and the substitution rate in $\{0,4,\dots,32\}\%$ at $K=32$. On eight paired layouts with $K=32$ and $20\%$ substitution, we vary the temperature in $\{0.8,0.9,\dots,1.2\}$ and the generator among OPT-1.3B, Qwen2.5-1.5B \cite{qwen2025}, Llama-3.2-1B \cite{grattafiori2024llama3} and Sheared-LLaMA-2.7B \cite{xia2024sheared}.

We compare calibrated Max-Pivot, the Local-Count weighting and BW of \cite{zhou2027icassp}, and BLADE. Throughout, Max-Pivot denotes the calibrated version. The first three use \eqref{eq:bw-rule} with uniform, two-level local and chain-based weights, respectively. For each mode, Local-Count counts the unmasked pivots above $0.95$ within five positions of $t$, excluding $t$. It assigns the weight $0.4+0.6/K$ to the mode with the largest count and $0.6/K$ to the others, breaking ties by the lowest index. BW uses $\rho_{\mathrm{sw}}=0.99$ and $\bar l=100$ as in \cite{zhou2027icassp}. BLADE estimates \eqref{eq:density-est} from the $N=20{,}718$ tokens of $56$ development and calibration sources. For each temperature and model, it re-estimates \eqref{eq:density-est} from $12$ matched sources. Its parameters $\theta=0.3$, $\rho_0=\rho_1=0.98$ and $\beta=0.1$ are selected on development documents and fixed before evaluation. BW and BLADE use $\varepsilon=10^{-3}$, and we set $\alpha=0.05$. The Oracle of Corollary~\ref{cor:benchmarks}(ii), given $S_t=j\ge1$, reports $j$ if $r_{t,j}>0.95$ and $0$ otherwise.

Correct-label power and the misattribution rate are the fractions of retained, unmasked positions assigned to the correct and to a wrong mode, and Type I error is the fraction of null positions detected. For edited text, rates are pooled over layouts within each substitution realization and averaged over the eight realizations. Shaded regions and error bars are $95\%$ percentile intervals from $2{,}000$ bootstrap resamples of layouts. Each resample keeps all eight realizations of a drawn layout and is shared by all detectors, which gives paired intervals for the gains.

\subsection{Type I Error}
\label{sec:exp-type-i}

Fig.~\ref{fig:type-i} tests calibration as the candidate set expands. Panel (a) uses human-written text, and panel (b) uses unretained positions in the edited documents, a harder null population because their neighbors still carry watermark evidence.

\begin{figure}[!t]
\centering
\includegraphics{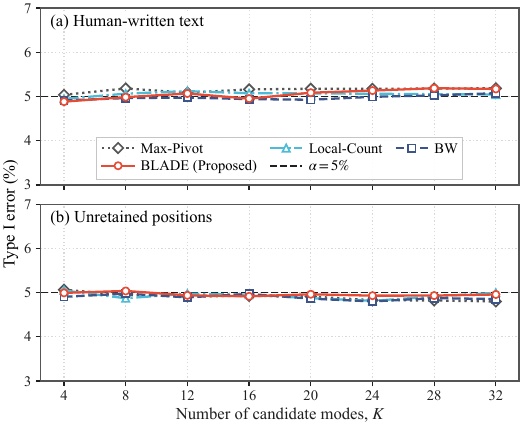}
\caption{Token-level Type I error versus $K$: (a) human-written text; (b) unretained positions.}
\label{fig:type-i}
\end{figure}

All four detectors stay near the nominal level in both populations. On human-written text, their rates lie between $4.88\%$ and $5.19\%$ over the tested candidate counts. On edited text, the mean rates range from $4.80\%$ to $5.06\%$, and BLADE's means lie between $4.92\%$ and $5.04\%$. Increasing the candidate count produces no systematic inflation in either population.

\subsection{Correct-Label Power}
\label{sec:exp-power}

\subsubsection{Candidate Count}
\label{sec:exp-power-k}

Fig.~\ref{fig:power-k} tests robustness to irrelevant candidates: the $20\%$ substitution documents and their four generating modes remain unchanged as $K$ increases, so every added key is a competitor without signal.

\begin{figure}[!t]
\centering
\includegraphics{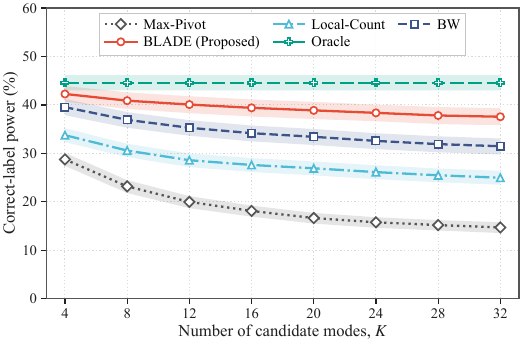}
\caption{Correct-label power versus $K$ at $20\%$ substitution. Shaded regions are $95\%$ bootstrap intervals (Section~\ref{sec:exp-metrics}).}
\label{fig:power-k}
\end{figure}

BLADE has the highest mean correct-label power among the practical detectors throughout the sweep, and the paired $95\%$ interval of its gain over BW excludes zero at every $K$. At $K=32$, Max-Pivot, Local-Count, BW and BLADE attain $14.7\%$, $25.0\%$, $31.5\%$ and $37.6\%$, respectively, and the misattribution rate is $1.89\%$ for BW and $1.18\%$ for BLADE, well below the bound of $5.13\%$ in Lemma~\ref{lem:misattribution}.

BLADE is also the least sensitive to added candidates. From $K=4$ to $K=32$, its power changes by $4.7$ percentage points, compared with $8.1$ for BW, so its gain over BW grows from $2.7$ to $6.1$ points. Max-Pivot's declining power follows the competition mechanism of Section~\ref{sec:form-maxpivot}. The weighted detectors reduce the loss by directing more budget toward modes supported by neighboring evidence.

The Oracle stays at $44.6\%$, since its pivot and threshold do not depend on $K$. Measured against the potential gain from mode information (Section~\ref{sec:opt-rule}), BLADE recovers $85\%$ of the gap between Max-Pivot and the Oracle at $K=4$ and $77\%$ at $K=32$, compared with $68\%$ and $56\%$ for BW.

\subsubsection{Substitution Rate}
\label{sec:exp-power-substitution}

Fig.~\ref{fig:power-substitution} varies the substitution rate at $K=32$ with all detector parameters unchanged. The rate controls how much neighboring evidence survives the edits.

\begin{figure}[!t]
\centering
\includegraphics{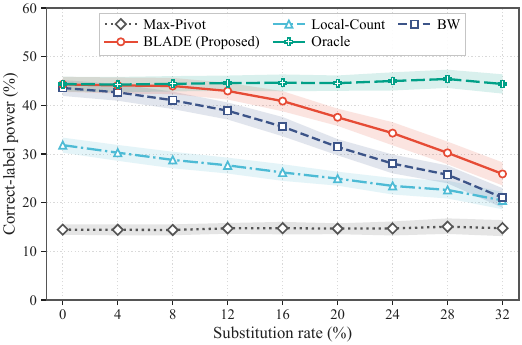}
\caption{Correct-label power versus substitution rate with $K=32$. Shaded regions are $95\%$ bootstrap intervals (Section~\ref{sec:exp-metrics}).}
\label{fig:power-substitution}
\end{figure}

\begin{figure}[!t]
\centering
\includegraphics{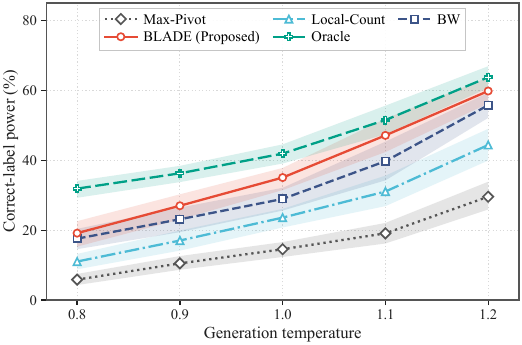}
\caption{Correct-label power versus generation temperature with $K=32$ and $20\%$ substitution. Shaded regions are $95\%$ bootstrap intervals (Section~\ref{sec:exp-metrics}).}
\label{fig:power-temperature}
\end{figure}

\begin{figure}[!t]
\centering
\includegraphics{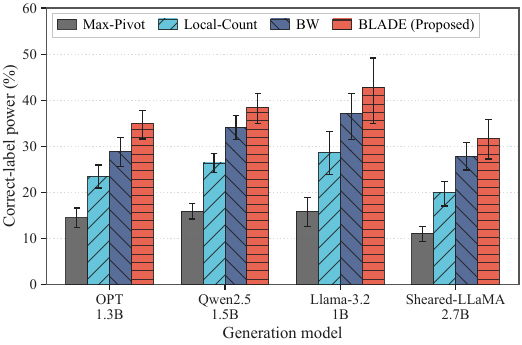}
\caption{Correct-label power across four generator models with $K=32$ and $20\%$ substitution. Error bars are $95\%$ bootstrap intervals (Section~\ref{sec:exp-metrics}).}
\label{fig:power-model}
\end{figure}

\begin{figure*}[!t]
\centering
\includegraphics{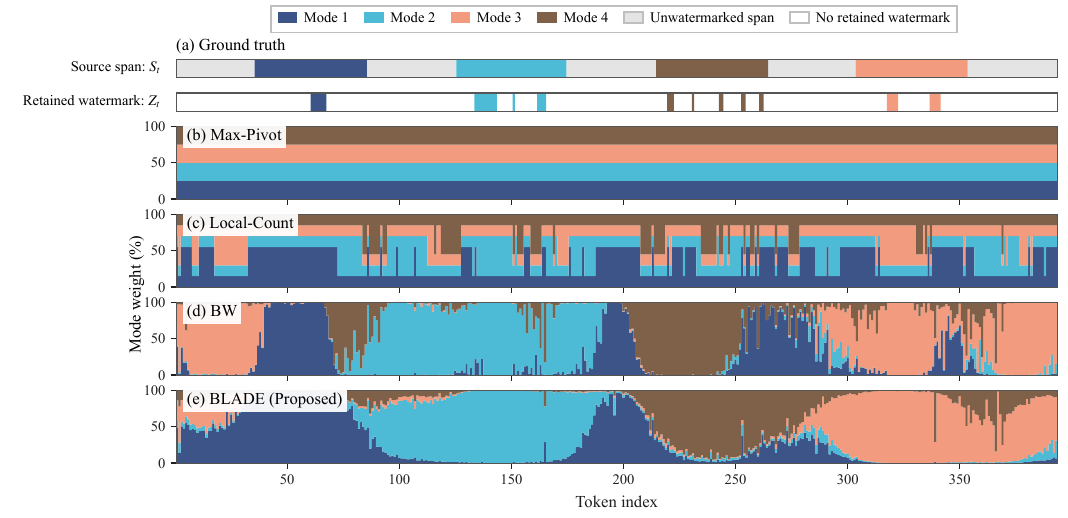}
\caption{Mode weights on a selected document with $K=4$ and $20\%$ substitution: (a) span labels $S_t$ and retained-token states $Z_t$; (b)--(e) stacked mode weights of the four detectors.}
\label{fig:mode-weights}
\end{figure*}

BLADE has the highest mean correct-label power among the practical detectors at all nine rates, and the paired $95\%$ interval of its gain over BW excludes zero at each rate. Without substitution, BLADE nearly matches the Oracle ($44.3\%$ versus $44.4\%$). As substitution removes neighboring evidence, its gain over BW grows from $0.7$ points to $6.3$ points at $24\%$ substitution and remains $4.9$ points at $32\%$.

Max-Pivot and the Oracle decide from the target pivots alone, which a retained six-token sequence leaves unchanged, so their correct-label powers are nearly flat. The weighted detectors draw their gain from neighboring support for the true mode, so their power tracks the amount of this support that survives the edits.

\subsubsection{Generation Temperature}
\label{sec:exp-power-temperature}

Fig.~\ref{fig:power-temperature} varies the signal strength through the generation temperature. The five conditions share the paired prompts, eight layouts and matched $12$-source references of Section~\ref{sec:exp-protocol}.

BLADE has the highest mean correct-label power among the practical detectors at every temperature, rising from $19.1\%$ to $59.8\%$, with gains over BW between $1.6$ and $7.4$ points. Mean power increases for all five methods across the sweep. The Oracle rises from $31.9\%$ to $63.8\%$, and BLADE closes in on it, with the gap narrowing from $12.7$ to $3.9$ points.

The common increase follows the signal mechanism of Section~\ref{sec:prelim-signal}: on the same source spans, raising the temperature lowers the median selected-token probability $p_t^\star$ from $0.343$ to $0.0254$ and raises the single-token power from $32.5\%$ to $64.1\%$ (Fig.~\ref{fig:single-mode}), within one percentage point of Corollary~\ref{cor:signal-law} evaluated at the recorded $p_t^\star$. Since the watermark preserves the NTP distribution, this gain comes from less predictable text.

\subsubsection{Generator Model}
\label{sec:exp-power-model}

Fig.~\ref{fig:power-model} tests whether the gain transfers across the four generator models of Section~\ref{sec:exp-protocol} at temperature $1$. Each model uses the eight paired layouts and a model-matched $12$-source reference, and the OPT bar is computed on the same sample.

BLADE has the highest mean correct-label power among the practical detectors for every tested model, gaining $3.7$--$6.1$ percentage points over BW without retuning its hyperparameters. Only the density estimate is re-estimated for each model.

\subsection{Mode-Weight Illustration}
\label{sec:exp-mode-weights}

Fig.~\ref{fig:mode-weights} compares the mode weights of the four detectors, i.e., their weights on the $K$ modes at each position, on one document with $20\%$ substitution and $K=4$. Panel (a) separates the source-span labels $S_t$ from retained-token states $Z_t$, since substitution can give $Z_t=0$ within a span with $S_t\neq0$. Max-Pivot has uniform weights, while Local-Count changes its preferred mode through local counts. BW often concentrates on one mode but exhibits sharp local changes. BLADE gives more sustained support within the watermarked spans, particularly Mode 2, with more gradual transitions near their boundaries.

BW and BLADE plot $\widehat{\mathbf q}_t$ from \eqref{eq:bw-prior} and \eqref{eq:blade-prior}. All four weight vectors sum to one even at null positions, so detection still depends on the target pivots. On the $42$ retained, unmasked tokens, the mean true-mode weight is $79.8\%$ for BW and $90.1\%$ for BLADE.

\section{Conclusion}
\label{sec:conclusion}

In this paper, we have considered token-level attribution of LLM watermarks among $K$ candidate modes, formulated as maximizing the prior-weighted correct-label power under a Type I error constraint. We have shown that the optimal rule thresholds the largest prior-weighted alternative density at one common threshold, which splits the Type I error budget among the modes and also bounds the misattribution rate. Based on this structure, we have proposed BLADE. It estimates the mode prior from the neighboring pivots with a working Markov chain that includes a null state, and the alternative density from a reference sample of token probabilities. 
We have also shown that its correct-label power approaches the optimum as the prior and density errors vanish. Numerical experiments on mixed documents have shown that BLADE attains the highest mean correct-label power among the practical detectors across candidate counts, substitution rates, generation temperatures and four generator models. Its Type I error stays near the nominal level. 

\appendices
\section{Proofs of Facts \ref{fact:max-pivot} and \ref{fact:cal-max-pivot}}
\label{app:C}

Let $F_{H,j}(x)=x^d$ for all $j$, and let $\tau$ be the threshold of Max-Pivot. Under $\mathcal H_0^{(t)}$, the $K$ pivots are independent $\operatorname{Unif}(0,1)$ variables by Assumption~\ref{ass:pseudorandom}, so the rule accepts the null with probability $\Pr(\max_ir_{t,i}\le\tau)=\tau^K$. Under $\mathcal H_j^{(t)}$, it reports mode $j$ if and only if $r_{t,j}$ exceeds both $\tau$ and the other $K-1$ pivots. Given $r_{t,j}=r$, the latter event has probability $r^{K-1}$, so
\begin{equation}
\label{eq:C-tau}
P_j(\delta_t)=d\int_\tau^1r^{K+d-2}\,dr=\frac{d}{K+d-1}\bigl(1-\tau^{K+d-1}\bigr).
\end{equation}

The factor $d/(K+d-1)=d/K+O(K^{-2})$ is the probability that the true pivot exceeds the other $K-1$ pivots.

\emph{Fact \ref{fact:max-pivot}.} With $\tau=1-\alpha$, the Type I error is $1-(1-\alpha)^K\to1$, and \eqref{eq:C-tau} gives
\begin{equation}
\label{eq:max-pivot-power}
P_j(\delta_t)=\frac{d}{K+d-1}\Bigl[1-(1-\alpha)^{K+d-1}\Bigr].
\end{equation}

The bracket is $1-O\bigl((1-\alpha)^K\bigr)$, so $P_j(\delta_t)=d/K+o(1/K)$.

\emph{Fact \ref{fact:cal-max-pivot}.} With $\tau=(1-\alpha)^{1/K}$, the Type I error is $1-\tau^K=\alpha$, and \eqref{eq:C-tau} gives
\begin{equation}
\label{eq:max-pivot-cal-power}
P_j(\delta_t)=\frac{d}{K+d-1}\Bigl[1-(1-\alpha)^{(K+d-1)/K}\Bigr].
\end{equation}

The bracket is $1-(1-\alpha)^{1+(d-1)/K}=\alpha+O(1/K)$, so $P_j(\delta_t)=\alpha d/K+o(1/K)$.

\section{Proof of Theorem~\ref{thm:optimal}}
\label{app:A}

In Appendices~\ref{app:A}--\ref{app:E}, the position $t$ and the side information $\mathcal R_{-t}$ are fixed as in Section~\ref{sec:prior-neighbors}, and $\mathbf r=(r_1,\dots,r_K)$ denotes a point of $[0,1]^K$, as in \eqref{eq:power-typeI}. We write $\operatorname{Leb}(\mathcal E)$ for the Lebesgue measure of a set $\mathcal E\subseteq[0,1]^K$. By \eqref{eq:power-typeI}, a rule $\delta_t$ with rejection region $\Omega=\{\mathbf r:\delta_t(\mathbf r)\neq0\}$ has $P_{\mathrm I}(\delta_t)=\operatorname{Leb}(\Omega)$.

\begin{IEEEproof}[Proof of Theorem~\ref{thm:optimal}]
\emph{Labeling.} For any rule $\delta_t$ with rejection region $\Omega$, \eqref{eq:power-typeI} gives
\begin{equation}
\label{eq:A-label}
\begin{aligned}
\sum_{j=1}^Kq_{t,j}P_j(\delta_t)
&=\int_{\Omega}\sum_{j=1}^Kq_{t,j}f_j(r_j)\,\mathbb I\{\delta_t(\mathbf r)=j\}\,d\mathbf r\\
&\le\int_{\Omega}g_t(\mathbf r)\,d\mathbf r,
\end{aligned}
\end{equation}
since at most one indicator is nonzero. Equality holds if and only if $\delta_t$ reports $\arg\max_jq_{t,j}f_j(r_j)$ almost everywhere on $\Omega$.

\emph{Rejection region.} The null CDF $A_t$ of $g_t(\mathbf r)$ in \eqref{eq:null-cdf} is continuous and satisfies $A_t(0)=0$. It is strictly increasing wherever $0<A_t(c)<1$, because there every factor $f_j^{-1}(c/q_{t,j})$ is positive and at least one lies in $(0,1)$, where it is strictly increasing. Hence, for $0<\alpha<1$, a unique $c_t>0$ satisfies $A_t(c_t)=1-\alpha$, and the region $\Omega^\star=\{g_t>c_t\}$ has $\operatorname{Leb}(\Omega^\star)=\alpha$. For any $\Omega$ with $\operatorname{Leb}(\Omega)\le\alpha$,
\begin{equation}
\label{eq:A-bathtub}
\begin{aligned}
&\int_{\Omega^\star}g_t-\int_{\Omega}g_t
=\int_{\Omega^\star\setminus\Omega}g_t-\int_{\Omega\setminus\Omega^\star}g_t\\
&\qquad\ge c_t\bigl[\operatorname{Leb}(\Omega^\star\setminus\Omega)-\operatorname{Leb}(\Omega\setminus\Omega^\star)\bigr]\ge0,
\end{aligned}
\end{equation}
because $g_t>c_t$ on $\Omega^\star$, $g_t\le c_t$ outside it, and $\operatorname{Leb}(\Omega^\star)\ge\operatorname{Leb}(\Omega)$. This is the bathtub principle \cite{lieb2001}, i.e., the Neyman--Pearson lemma with $g_t$ in place of an alternative density. By \eqref{eq:A-label} and \eqref{eq:A-bathtub}, $\delta_t^\star$ attains the maximum of \eqref{eq:problem}.

\emph{Union of tails.} By \eqref{eq:private-level} and the monotonicity of $f_j$, $\{q_{t,j}f_j(r_j)>c_t\}=\{r_j>1-s_{t,j}\}$. Thus $\Omega^\star=\bigcup_j\{r_j>1-s_{t,j}\}$, whose complement is the box $\prod_j[0,1-s_{t,j}]$. This gives \eqref{eq:balance}.
\end{IEEEproof}

\section{Proof of Proposition~\ref{prop:blade-loss} and Properties of the Density Term}
\label{app:B}

\begin{IEEEproof}[Proof of Proposition~\ref{prop:blade-loss}]
Let $\delta_t^\star$ be the optimal rule \eqref{eq:opt-rule} and $\delta_t$ the BLADE rule \eqref{eq:blade-rule} with $c_t$ solving \eqref{eq:balance} exactly. Both rules are feasible for \eqref{eq:problem}, since both are calibrated under the same conditional uniform null distribution (Lemma~\ref{lem:calibration}). Hence the loss in \eqref{eq:blade-loss} is nonnegative.

For $\eta\in\{\delta_t^\star,\delta_t\}$, let $\nu_j^\eta(x)$ be the probability that $\eta$ reports mode $j$ when $r_j=x$ and the other $K-1$ pivots are independent $\operatorname{Unif}(0,1)$ variables,
\begin{equation}
\label{eq:B-nu}
\nu_j^\eta(x)=\int_{[0,1]^{K-1}}\mathbb I\{\eta(x,\mathbf r_{-j})=j\}\,d\mathbf r_{-j},
\end{equation}
where $(x,\mathbf r_{-j})$ denotes the point with $r_j=x$. By Corollary~\ref{cor:product}(b), $P_j(\eta)=\int_0^1\nu_j^\eta(x)f_j(x)\,dx$. Both rules report the largest of coordinatewise nondecreasing scores above a fixed threshold, so increasing $r_j$ preserves a report of mode $j$. Since ties are null, each $\nu_j^\eta$ is nondecreasing with values in $[0,1]$. For $x\le y$, the increments $\nu_j^\eta(y)-\nu_j^\eta(x)$ of both rules lie in $[0,1]$, so $\Delta_j=\nu_j^{\delta_t^\star}-\nu_j^{\delta_t}$ satisfies
\begin{equation}
\label{eq:B-osc}
\sup_x\Delta_j(x)-\inf_x\Delta_j(x)\le1,\qquad\lvert\Delta_j(x)\rvert\le1.
\end{equation}

Since $\widehat{\mathbf q}_t>\mathbf 0$ and $\widehat f$ is continuous and strictly increasing, Theorem~\ref{thm:optimal} applied to the estimated inputs shows that $\delta_t$ maximizes the weighted correct-label power computed with $\widehat{\mathbf q}_t$ and $\widehat f$. Hence $\sum_j\widehat q_{t,j}\int_0^1\Delta_j\widehat f\le0$, and subtracting this quantity gives
\begin{equation}
\label{eq:B-split}
\begin{aligned}
\mathrm{OPT}-\sum_jq_{t,j}P_j(\delta_t)
&=\sum_jq_{t,j}\int_0^1\Delta_jf_j\\
&\le\sum_jq_{t,j}\int_0^1\Delta_j\bigl(f_j-\widehat f\bigr)\\
&\quad+\sum_j\bigl(q_{t,j}-\widehat q_{t,j}\bigr)\int_0^1\Delta_j\widehat f.
\end{aligned}
\end{equation}

Let $m_j$ be the midpoint of $\sup_x\Delta_j(x)$ and $\inf_x\Delta_j(x)$, so that $\lvert\Delta_j-m_j\rvert\le1/2$ by \eqref{eq:B-osc}. Since $f_j$ and $\widehat f$ both integrate to one,
\begin{equation}
\label{eq:B-density}
\begin{aligned}
\Bigl\lvert\int_0^1\Delta_j\bigl(f_j-\widehat f\bigr)\Bigr\rvert
&=\Bigl\lvert\int_0^1(\Delta_j-m_j)\bigl(f_j-\widehat f\bigr)\Bigr\rvert\\
&\le\frac12\lVert f_j-\widehat f\rVert_{L^1}.
\end{aligned}
\end{equation}

Moreover, $\lvert\int_0^1\Delta_j\widehat f\rvert\le1$, since $\lvert\Delta_j\rvert\le1$ and $\widehat f$ integrates to one. Applying these bounds to the two sums in \eqref{eq:B-split} gives \eqref{eq:blade-loss}.
\end{IEEEproof}

\emph{Consistency of the density estimate.} Let the reference sample be drawn i.i.d.\ from the distribution underlying $f_{\mathrm{ref}}$ in \eqref{eq:density-split}. Since $\mathbb E[d_i]<\infty$, the strong law of large numbers gives $\widehat f(x)\to f_{\mathrm{ref}}(x)$ almost surely for each $x\in[0,1]$. By Fubini's theorem, with probability one the convergence holds for almost every $x$, and Scheff\'e's lemma gives $\lVert f_{\mathrm{ref}}-\widehat f\rVert_{L^1}\to0$, since both functions are probability densities.

\emph{Sharpness of the density coefficient.} Consider the example of Section~\ref{sec:analysis-blade} with a vocabulary of size $\lvert\mathcal V\rvert\ge D$ for each $D$. Let $\mathrm{OPT}_D$ denote its optimum, and let $T_2(u)=1-(1-u)^2$ be the power of the mode-2 tail at level $u$. Any feasible rule reports mode 2 on a set of null volume at most $\alpha$. Since $f_2$ is increasing, the bathtub principle in \eqref{eq:A-bathtub} gives $P_2\le T_2(\alpha)$, and $P_1\le1$. Hence
\begin{equation}
\label{eq:B-upper}
\mathrm{OPT}_D\le0.4+0.6\,T_2(\alpha)=0.4585.
\end{equation}

Conversely, for large $D$, let $s_D=\log D/D<\alpha$ and $b_D=(\alpha-s_D)/(1-s_D)$. The rule that reports mode 1 if $r_1>1-s_D$, and otherwise mode 2 if $r_2>1-b_D$, has null rejection probability $s_D+(1-s_D)b_D=\alpha$ and weighted correct-label power
\begin{equation}
\label{eq:B-lower}
0.4\bigl[1-(1-s_D)^D\bigr]+0.6\,(1-s_D)\,T_2(b_D)\longrightarrow0.4585,
\end{equation}
since $(1-s_D)^D\le1/D$ and $b_D\to\alpha$. Thus $\mathrm{OPT}_D\to0.4585$, whereas rule \eqref{eq:blade-rule} has power $0.0585$, so its loss tends to $0.4$. For $D>2$, $f_1$ and $f_2$ cross at $a_D=(2/D)^{1/(D-2)}$, and
\begin{equation}
\label{eq:B-l1}
\frac12\lVert f_1-f_2\rVert_{L^1}=\int_0^{a_D}(f_2-f_1)=a_D^2\Bigl(1-\frac2D\Bigr)\longrightarrow1.
\end{equation}

The density term $\frac12q_{t,1}\lVert f_1-\widehat f\rVert_{L^1}$ in \eqref{eq:blade-loss} therefore also tends to $0.4$.

\section{Explicit Threshold Computation}
\label{app:E}
\label{app:E-threshold}

Algorithm~\ref{alg:blade} evaluates the calibrated BLADE decision directly by Lemma~\ref{lem:direct-eval}. When the threshold $c_t$ or the per-mode levels are needed explicitly, they follow from the calibration equation \eqref{eq:balance} with $\widehat{\mathbf q}_t$ and $\widehat f$ in place of $\mathbf q_t$ and $f_j$. By \eqref{eq:blade-prior}, $\widehat q_{t,j}>0$ for every $j$, and $\widehat f$ is continuous and strictly increasing whenever some $d_i>1$. With the clipped inverse of Section~\ref{sec:opt-calibration}, a candidate threshold $c$ gives
\begin{equation}
\label{eq:E-levels}
\begin{aligned}
s_{t,j}(c)&=1-\widehat f^{-1}\Bigl(\frac{c}{\widehat q_{t,j}}\Bigr),\\
A_t(c)&=\prod_{j=1}^K\bigl(1-s_{t,j}(c)\bigr).
\end{aligned}
\end{equation}

As in Appendix~\ref{app:A}, $A_t$ is continuous, nondecreasing and strictly increasing wherever $0<A_t(c)<1$. The endpoints
\begin{equation}
\label{eq:E-bracket}
c_{\mathrm L}=\min_j\widehat q_{t,j}\widehat f(0),\qquad
c_{\mathrm U}=\max_j\widehat q_{t,j}\widehat f(1)
\end{equation}
satisfy $A_t(c_{\mathrm L})=0$ and $A_t(c_{\mathrm U})=1$, so $A_t(c_t)=1-\alpha$ has a unique solution in $[c_{\mathrm L},c_{\mathrm U}]$ for $0<\alpha<1$. Bisection keeps $A_t(c_{\mathrm L})\le1-\alpha\le A_t(c_{\mathrm U})$, replacing the lower endpoint when the midpoint value is below $1-\alpha$ and the upper endpoint otherwise, and stops when
\begin{equation}
\label{eq:E-stop}
0\le A_t(c_{\mathrm U})-(1-\alpha)\le\epsilon_{\mathrm{tol}}.
\end{equation}

By Corollary~\ref{cor:product}(a), the rule with threshold $c_{\mathrm U}$ has conditional Type I error $1-A_t(c_{\mathrm U})\in[\alpha-\epsilon_{\mathrm{tol}},\alpha]$ when $A_t$ is evaluated exactly. This guarantee requires the exclusion of the target pivots from $\widehat{\mathbf q}_t$, but not accurate estimates of $\mathbf q_t$ or $f_j$.

\emph{Implementation in the experiments.} The experiments floor reference probabilities at $10^{-12}$ and tabulate $\log\widehat f$ on a fixed grid in $y=-\log(1-r)$, with a density floor of $10^{-300}$. Density queries interpolate this table, and inverse queries interpolate the same table with its coordinates exchanged. Forward queries cap $r$ at $1-10^{-12}$, and inverse queries use $y\in[0,40]$. Bisection runs for $60$ steps in $\log c$ on the interval $[L_{\min}+m_t-1,\,L_{\max}+m_t+1]$, where $L_{\min}$ and $L_{\max}$ are the extreme tabulated log-densities and $m_t=\max_j\log\widehat q_{t,j}$. This interval replaces $[\log c_{\mathrm L},\log c_{\mathrm U}]$ because $c_{\mathrm L}$ can be zero. The table has $G=4{,}001$ nodes, and the decision uses the upper endpoint of the final interval. The reported BLADE results use this implementation. The decision rejects when $r_{t,j}>1-s_{t,j}(c_{\mathrm U})$ for some $j$, with the levels read from the table. The bisection keeps $\prod_j\bigl(1-s_{t,j}(c_{\mathrm U})\bigr)\ge1-\alpha$. The interpolation and the caps therefore change only the levels, and the conditional Type I error stays at most $\alpha$.

\bibliographystyle{IEEEtran}
\bibliography{references}
\end{document}